\documentclass[journal]{IEEEtran}

\ifCLASSINFOpdf
\else
\fi

\usepackage{cite} 
\usepackage{hyperref}
\usepackage{graphicx}
\usepackage[justification=centering]{caption}

\usepackage{algorithm}
\usepackage{algorithmic}
\usepackage{mathtools,amssymb}
\usepackage{color,soul}

\usepackage{amsmath}
\usepackage{amsthm}

\usepackage{subcaption}

\begin{document}
%
\title{A Unified Learning-based Optimization Framework for 0-1 Mixed Problems in Wireless Networks}
%
%
%

\author{        Kairong Ma,~\IEEEmembership{Graduate Student~Member,~IEEE},
				Yao Sun,~\IEEEmembership{Senior~Member,~IEEE}, 
                Shuheng Hua,~\IEEEmembership{Graduate Student~Member,~IEEE}, 
                Muhammad Ali Imran,~\IEEEmembership{Fellow,~IEEE}, 
                and Walid Saad,~\IEEEmembership{Fellow,~IEEE}

                 \thanks{
	
	Kairong Ma, Yao Sun, Shuheng Hua and Muhammad Ali Imran are with the James Watt School of Engineering, University of Glasgow, Glasgow G12 8QQ, UK (e-mail: \{k.ma.2, s.hua.1\}@research.gla.ac.uk; \{Yao.Sun, Muhammad.Imran\}@glasgow.ac.uk.)
 
Walid Saad is with the Bradley Department of Electrical and Computer Engineering, Virginia Tech, Arlington, VA, 22203, USA, (e-mail: walids@vt.edu.

The research work of Kairong Ma, Yao Sun, Shuheng Hua, and Muhammad Ali Imran was supported by EPSRC projects, CHEDDAR EP/X040518/1 and CHEDDAR uplift EP/Y037421/1.

© 2025 IEEE. Personal use is permitted. For any other purposes, permission must be obtained from the IEEE. Accepted for publication in IEEE Transactions on Communications. DOI: 10.1109/TCOMM.2025.3618171. 
} }

\maketitle

\begin{abstract}
Several wireless networking problems are often posed as  0-1 mixed optimization problems, which involve binary variables (e.g., selection of access points, channels, and tasks) and continuous variables (e.g., allocation of bandwidth, power, and computing resources). 
Traditional optimization methods as well as reinforcement learning (RL) algorithms have been widely exploited to solve these problems under different network scenarios. 
However, solving such problems becomes more challenging when dealing with a large network scale, multi-dimensional radio resources, and diversified service requirements. 
To this end, in this paper, a unified framework that combines RL and optimization theory is proposed to solve 0-1 mixed optimization problems in wireless networks. 
First, RL is used to capture the process of solving binary variables as a sequential decision-making task. 
During the decision-making steps, the binary (0-1) variables are relaxed and, then, a relaxed problem is solved to obtain a relaxed solution, which serves as prior information to guide RL searching policy.
Then, at the end of decision-making process, the search policy is updated via suboptimal objective value based on decisions made. 
The performance bound and convergence guarantees of the proposed framework are then proven theoretically. 
An extension of this approach is provided to solve problems with a non-convex objective function and/or non-convex constraints. 
Numerical results show that the proposed approach reduces the convergence time by about 30\% over B\&B in small-scale problems with slightly higher objective values. 
In large-scale scenarios, it can improve the normalized objective values by 20\% over RL with a shorter convergence time.

\end{abstract}


\IEEEpeerreviewmaketitle

\newtheorem{definition}{\textbf{Definition}}
\newtheorem{lemma}{\textbf{Lemma}}
\newtheorem{proposition}{\textbf{Proposition}}
\makeatletter
\renewenvironment{proof}[1][\proofname]{\par
  \pushQED{\qed}\normalfont\topsep6\p@\@plus6\p@\relax
  \trivlist\item[\hskip\labelsep
    \itshape
    #1\@addpunct{:}]\ignorespaces
}{%
  \popQED\endtrivlist\@endpefalse
}
\makeatother

%
%
%
%

\section{Introduction}
In wireless communication networks, the use of 0-1 mixed optimization problems is a common approach to address resource management challenges. 
This type of problems involves both binary variables and continuous variables, 
with an optimization objective defined based on a network performance metric and several constraints on resource budgets and service requirements. 
In wireless networks, binary variables are typically used to decide the selection of access points, channels, resource blocks, or tasks. 
Continuous variables are typically used to allocate the amount of resources, such as bandwidth or energy. 
Network operators aim to optimally allocate the limited wireless resources 
to achieve one or more network optimization goals, such as system throughput, spectrum efficiency, or network capacity.


\textcolor{black}{
However, from the mathematical perspective, solving 0-1 mixed optimization problems is challenging. 
 The challenges stem from the binary variables, the coupling of binary and continuous variables, and often, non-convex objective functions and/or constraints.}
A variety of optimization-based methods such as linear programming relaxation \cite{ref23}, branch and cut \cite{ref24}, and dynamic programming \cite{ref25} have been used to solve these problems.
 These methods can provide the exactly optimal solution for small-scale problems, 
but they cannot be directly applied to dealing with general large-scale problems due to non-convexity and the problem's combinatorial nature.
To overcome this challenge, wireless networks may rely on heuristic algorithms and machine learning (ML) methods to generate approximate solutions \cite{ref90}. 
\textcolor{black}{For example, some deep reinforcement learning (DRL) approaches now use large reasoning models to enhance performance \cite{DU1}. Others focus on developing novel, diffusion-based reward shaping schemes \cite{DU2}. In parallel, distributed frameworks like federated learning are also gaining traction. Researchers are designing long-term contribution incentives in these systems \cite{RR1}. They are also building robust learning frameworks that use verifiable perturbations \cite{RR2}. While these approaches provide reasonable solutions for large-scale problems, they cannot guarantee optimality or convergence \cite{ref100}.
}

\vspace{-5pt}
\subsection{{Related Works}}
Prior works \cite{optimization_based,heuristic_based,ref7,ref13,ref14,ref15,ref16,refHungarian,walid1,walid2,walid3,DRL,ref17,ref20,RL_02,genetic,heuristic_greedy,heuristic_01} on wireless resource management leverage optimization techniques that can be classified into three categories:
traditional optimization-based algorithms \cite{optimization_based}, heuristic algorithms \cite{heuristic_based}, and more recently ML-based algorithms \cite{ref7}.
In general, traditional optimization methods, such as branch and bound (B\&B), cutting planes, and dynamic programming, can guarantee an optimal solution. 
However, these methods are not suitable for large-scale problems due to their exponential complexity. 
The authors in \cite{ref13} integrate Balas's method with B\&B search strategies, including depth-first, breadth-first, and best-first search,
to reduce computational demands for optimizing the binary integer programming problems. 
However, memory requirement and computing complexity remain a challenge for large-scale problems because the B\&B methods need to generate a list of nodes for evaluation.
The authors in \cite{ref14} and \cite{ref15} use a linear relaxation on binary variables, and obtain fractional solutions by solving the relaxed problem.
A direct rounding method is used in \cite{ref14} to recover the binary variables. 
The authors in \cite{ref15} take the elements of the largest fractional solution of each row to 1 and the rest to 0, because there is a special constraint that the binary variables of each row sum to 1. 
The authors in \cite{ref16} use a penalty factor in a suboptimal algorithm to enforce binary constraints on subcarrier allocation in a 0-1 mixed programming framework. 
This penalty factor penalizes deviations from binary values (0 or 1), ensuring that allocation decisions converge to binary values during the optimization process.
In addition, the Hungarian method is also used in \cite{refHungarian} to address combinatorial assignment problems in a polynomial time. 
However, the Hungarian method cannot effectively handle 0-1 mixed problems due to the coupling of binary and continuous variables.

Meanwhile, the works in \cite{walid1,walid2,walid3,DRL,ref17,ref20,RL_02} adopt ML solutions for solving 0-1 mixed optimization problems in various wireless scenarios. 
In general, ML techniques are used in two ways to solve these problems: using neural networks to map inputs (problem parameters) and outputs (solutions), or employing reinforcement learning (RL) to explore solution space.
In \cite{walid1}, the authors employ an RL-based curriculum learning framework, which incrementally identifies hierarchical belief structures and refines semantic event descriptions to optimize task execution efficiency, communication costs, and belief utilization.
The authors in \cite{walid2} use an ML solution by introducing a federated echo state network to predict the locations and orientations of virtual reality (VR) users in a distributed manner, optimizing the user-base station association to minimize breaks in presence events in wireless networks.
\textcolor{black}{The work of \cite{walid3} uses a neural combinatorial deep reinforcement learning (NCRL) framework to optimize strategies in wireless networks. Their approach combines convex optimization with a state representation based on long short-term memory (LSTM). The goal is to jointly optimize the trajectory and scheduling for UAVs.}
The works in \cite{DRL} and \cite{ref17} introduce deep learning (DL)-based frameworks for optimization in wireless networks. 
In \cite{DRL}, binary offloading decisions are separated from resource allocation tasks, using a DNN-based module for the binary decisions and an optimization-based module for resource allocation. 
Similarly, the work in \cite{ref17} proposes to approximate solutions for non-convex constrained optimization problems, using binarization and a primal-dual training method.
The authors in \cite{ref20} integrate imitation learning with the B\&B algorithm to optimize pruning policies, achieving near-optimal performance.
In \cite{RL_02}, the authors employ RL in a game-theoretic framework to optimize energy efficiency through iterative updates of user association and resource scheduling strategies.
In summary, DNN-based methods like those in \cite{walid1,walid2,walid3,DRL,ref17,ref20,RL_02} are useful for generating fast approximate solutions with low computational intensity. 
However, the quality of the DNN output depends on the robustness of the dataset and the effectiveness of the training process. 
RL-based methods can search the solution space based on historical data, however, they can be inefficient for problems with large solution spaces, 
particularly in the context of large-scale  wireless network.

Furthermore, the works in \cite{genetic,heuristic_greedy,heuristic_01} use heuristic methods to obtain acceptable solutions with significantly reduced computational complexity, 
despite these methods may lack guarantees of convergence and optimality. 
The authors in \cite{genetic} conduct a survey of the applications of genetic algorithms (GAs) in wireless networks. This paper covers how GAs have been utilized to address problems characterized by large, complex search spaces, particularly the 0-1 optimization problems in wireless networks. 
The authors in \cite{heuristic_greedy} develop a heuristic algorithm that combines a greedy approach for constructing initial feasible solutions and a local search method for exploring neighboring solutions to optimize binary variables.
The authors in \cite{heuristic_01} present a fast heuristic algorithm for channel assignment in wireless networks, which utilizes weighted maximal independent sets within a conflict graph to manage cumulative interference, enhancing computational speed and optimality.
Overall, heuristic algorithms such as those in \cite{genetic,heuristic_greedy,heuristic_01} face a tradeoff between solution optimality and computational complexity, i.e., they target for fast problem-solving with compromises on optimality.
However, these algorithms are typically tailored to specific problem structures, and their performance is highly sensitive to parameters (such as termination criteria or neighborhood size), and it often requires extensive experiments to find the optimal configurations. 
Additionally, the lack of performance guarantees results in inconsistent outcomes across different instances of the same problem, making these approaches less reliable for applications with high precision requirements.

\vspace{-5pt}
\subsection{{Contributions}}
The main contribution of this paper is to overcome the limitations of prior works by developing a unified approach for solving 0-1 mixed optimization problems in wireless networks.
\textcolor{black}{
While pure RL methods explore the solution space based on trial-and-error, which is inefficient in the vast solution spaces of large-scale wireless networks. 
The core novelty of our work is the creation of a framework that transforms this process into an informed search. We achieve this by creating a unique synergy between convex optimization and RL, where the former provides a prior to guide the exploration of the latter.} 
We first model the process of solving binary variables as an Markov decision process (MDP). Then, we relax binary variables to reformulate the original problem into a convex framework, where a feasible relaxed solution is derived based on convex optimization theory. 
Using the relaxed solution as prior information, RL is exploited to determine the suboptimal solution of the binary variables. 
When the RL algorithm reaches a terminal state, a binary solution is determined, and the continuous variables can be easily optimized using traditional optimization method.
In each episode, the objective value serves as the reward for RL, i.e., it is used to update the search policy. 
This process is repeated until the termination condition for the whole problem is met.
This proposed approach can enhance the efficiency of searching binary variables space as well as the likelihood of converging to a suboptimal solution. 
In summary, our key contributions include: 
\begin{itemize}
    \item 
    Considering the problem scalability and time complexity requirement in wireless networks, 
    we integrate convex optimization with RL to tackle the binary decision of 0–1 mixed problems.
    We solve a relaxed sub-problem before binary decision making, and use the relaxed solution as prior information to guide RL search policy.
    \item 
    We theoretically prove that the neighbourhood of the relaxed solution has a higher probability of containing suboptimal solutions. Thus, relaxed solutions can provide prior information for RL to guide its searching policy effectively.
    
    \item We discuss the extensions of our approach to handle non-convex objective functions and non-convex constraints. 
    These extensions broaden the applicability of our approach to complex real-world wireless networking scenarios.
    \item We conduct simulations to validate the proposed learning-based optimization approach for wireless problems with convex and non-convex objective functions.
    The simulations demonstrate that for large-scale wireless problems, 
    our proposed approach is more effective in exploring better solutions than traditional optimization methods and pure RL.
\end{itemize}

\textcolor{black}{The rest of this paper is organized as follows.} 
The general problem is presented in Section \uppercase\expandafter{\romannumeral 2}. 
Section \uppercase\expandafter{\romannumeral 3} presents the proposed learning-based optimization approach that is used solve 0-1 mixed problem is proposed. 
Section \uppercase\expandafter{\romannumeral 4} discusses the extensions of proposed approach in non-convex problems. 
Section \uppercase\expandafter{\romannumeral 5} presents our numerical results and conclusions are drawn in Section \uppercase\expandafter{\romannumeral 6}.

\section{A General 0-1 Mixed Optimization Problem for Wireless Networking}
{In a typical 0-1 mixed optimization problem in wireless} networks, binary variables are used to capture discrete decision variables, such as user association, channel assignment, or task allocation; while continuous variables are introduced to manage wireless resources, such as power, bandwidth, or time. 
The objective is generally to optimize network performance or resource utilization, subject to constraints on the quality-of-service (QoS) requirements as well as resource limitations. 
In this section, we formulate a general 0-1 mixed optimization problem, and provide several common examples from wireless networks including bandwidth optimization, power allocation, and task offloading.

\vspace{-10pt}
\subsection{{Problem Description and Formulation}}


Consider a scenario with \(N\) users and \(M\) base stations (BSs) in a wireless network in which a general 0-1 mixed optimization problem for network resource management can be formulated as follows: 
\begin{align} \label{P1}
\quad\max_{\boldsymbol{x,y}} \quad &f\left(\boldsymbol{x,y}\right)\\
\text{s.t.}\quad  &g_l(\boldsymbol{x,y}) \leq 0 ,\quad  \forall{l\in \mathcal{L}} ;\tag{\ref{P1}{a}}\\
&\textstyle{\sum_{j=1}^{M}}x_{ij}=1 , \quad \forall{i\in \mathcal{N}};\tag{\ref{P1}{b}}\\
&x_{ij}\in \{0,1 \} ,\quad \forall{i\in \mathcal{N}} , \forall{j\in \mathcal{M}}\tag{\ref{P1}{c}}.
\end{align}

\textit{Binary Variables:}
\(\boldsymbol{x}\in \{0,1\}^{N\times M}\) is a matrix of binary decision variables, where each \(x_{ij}\in \{0, 1\}\) indicates a yes/no, on/off, or included/excluded decision.\footnotemark 
\footnotetext{Our approach can still work under multiple sets of binary variables \(\boldsymbol{x}\in \{0,1\}^{N\times M\times K}\) or continuous variables \(\boldsymbol{y}\in {{\mathbb R}_0^+}^{N\times M\times K}\). For example, if we consider a joint optimization of bandwidth and power allocation with continuous variables, our approach can still work.}
For example, \(x_{ij}\) might represent whether or not to establish a connection between user \(i\) and BS \(j\), 
where \(x_{ij}=1\) if user \(i\) is connected to BS \(j\), and \(x_{ij}=0\) otherwise. 

\textit{Continuous Variables:} 
\(\boldsymbol{y}\in {{\mathbb R}_0^+}^{N\times M}\) is a matrix that consists of continuous variables.
It could represent the amount of resources allocated to links or nodes, such as bandwidth or power. 
For example, \(y_{ij}\) can be the transmit power at BS \(j\) to user \(i\), 
or the bandwidth allocated from BS \(j\) to user \(i\).
These variables usually take on positive real values, and subject to resource budgets.

\textit{Objective Function:} 
The objective function \(f(\boldsymbol{x,y}): \left(\{0,1\}^{N\times M},{{\mathbb R}_0^+}^{N\times M} \right) \to \mathbb R\) is a wireless network performance metric that depends on both the binary variables \(\boldsymbol{x}\) and continuous variables \(\boldsymbol{y}\).
This function can capture the networks goal of maximizing the overall network throughput, minimizing the total power consumption, or optimizing any other relevant network performance metrics.
It should be noted that multi-objective optimization is out of the scope of our work. 

\textit{Constraints \(g_l(\boldsymbol{x,y}) \leq 0\):} 
A set of constraints is imposed by the network in order to ensure that the minimum QoS requirements are met within the given resource budget. 
For example, the sum of bandwidth allocated by BS \(j\) to its associated users should not exceed its capacity, i.e., \(\sum_i^N x_{ij}y_{ij}\leq B_j\), if we use \(x_{ij}\) and \(y_{ij}\) to represent cell association and bandwidth allocation, respectively. 
Meanwhile, each user might have its own requirements on service quality, in terms of throughput, latency, data rate, etc.

\textit{Exclusivity Constraint \(\sum_{j=1}^{M}x_{ij}=1\):} 
The exclusivity constraint is commonly needed in wireless networking optimization problems. 
It ensures that exactly one of the options or resources is selected from a set of possible choices for each instance \(i\). 
This constraint is required typically for BS association, caching, channel allocation, time slot allocation, etc. 
For example, for user association, this constraint ensures that each user is uniquely associated to exactly one BS to access networks. 

Problem (\ref{P1}) can be viewed as a general resource management problem for a wireless network with \(M\) BSs and \(N\) users. This problem is clearly formulated as a 0-1 mixed optimization problem. The goal is to optimize the objective function \(f(\boldsymbol{x,y})\) subject to \(L\) constrains \(g_l(\boldsymbol{x,y})\leq 0\) and an exclusivity constraint \(\sum_{j=1}^M x_{ij}=1\). 
\textcolor{black}{
Note that, this framework can be extended to problems with higher-dimensional variables, such as \(x\in \{0,1\}^{N\times M\times K}\). This generalization is possible because the problem can still be framed as a sequence of \(N\) decisions.
The fundamental principle of using a relaxed solution to inform the RL policy is independent of the variable dimension.
}

\vspace{-6pt}
\subsection{Examples of 0-1 Mixed Problems}
Problem (1) can model many wireless network resource management problems. Here, we provide three concrete examples in order to shed light on the applicability of this problem. 

\textbf{Bandwidth Allocation:}
Problem (1) can be naturally used to capture bandwidth allocation problems in which the goal is to determine the amount of bandwidth that each BS must allocate to its users, along with determining the association between users and BSs. 
Here, the variables \(x_{ij}\) and \(y_{ij}\) in problem (\ref{P1}) can represent BS association (denoted \(x_{ij}^{B}\)) and bandwidth allocation \(y_{ij}^{B}\) in these problems. 
The objective function \(f(\boldsymbol{x,y})\) could be throughput, spectrum efficiency, user satisfaction, fairness, latency, etc.
Taking the example of maximizing throughput, \(f_B(\boldsymbol{x}^B,\boldsymbol{y}^B)={\textstyle \sum_{i=1}^{N}}{\textstyle \sum_{j=1}^{M}} x_{ij}^{B}y_{ij}^{B}\log_2(1+\gamma_{ij})\), 
where \(\gamma_{ij}\) refers to the signal-to-noise ratio of the signal from user \(i\) at BS \(j\).
Hence, a typical bandwidth allocation problem can be formulated as follows: 
\begin{equation}
\begin{split}
&\max_{\boldsymbol{x}^B,\boldsymbol{y}^B} {\textstyle \sum_{i=1}^{N}}{\textstyle \sum_{j=1}^{M}} x_{ij}^B y_{ij}^B \log_2(1+\gamma_{ij})\\
&\text{s.t.}\left\{\begin{array}{lc}
{\textstyle \sum_{i=1}^{N}} x_{ij}^B y_{ij}^B \leq W_j \,\, ,\,\,\qquad\qquad\qquad  \forall {j}\in \mathcal{M};\\
\sum_{j=1}^{M}x_{ij}^B=1\,\, , \,\,\qquad\qquad\qquad\qquad \forall{i}\in \mathcal{N};\\
\sum_{j=1}^{M}x_{ij}^B y_{ij}^B \log_2(1+\gamma_{ij})\ge D_i,\quad \forall{i}\in \mathcal{N};\\
x_{ij}^B\in \left \{0,1\right \}\,\, , \qquad\qquad\qquad \forall{i}\in \mathcal{N},\forall{j}\in \mathcal{M};
\end{array}\right.
\end{split}
\end{equation}
where \(W_j\) is the total available bandwidth to be allocated by BS \(j\), 
and \(D_i\) is the minimum requirement of data rate for each user \(i\).

\textbf{Task Offloading:}
In task offloading problems, the goal is to optimally offload computational tasks from resource-constrained mobile terminals to more powerful remote servers. 
The goal in such problems is usually to  minimize latency, energy consumption, or computing costs, etc. 
Here, the binary variables can be mapped as task offloading decisions \(x_{ij}^T\), while continues variables can be computing resource allocation \(y_{ij}^{T1}\) and amount of computing task to offload \(y_{ij}^{T2}\).
Taking the example of minimizing energy consumption, the objective function can be \(f_T\left(\boldsymbol{x}^T,\boldsymbol{y}^{T1},\boldsymbol{y}^{T2}\right)={\textstyle \sum_{i=1}^{N}}{\textstyle \sum_{j=1}^{M}} \Big( \alpha E_{ij}^L\big((1- x_{ij}^T),y_{ij}^{T1},y_{ij}^{T2}\big) +\beta E_{ij}^O \big(x_{ij}^T,y_{ij}^{T1},y_{ij}^{T2}\big) \Big)\), 
and a typical task offloading problem can be formulated as: 
\begin{equation}
\begin{split}
&\min_{\boldsymbol{x}^T,\boldsymbol{y}^T} \,\,\, \ {\textstyle \sum_{i=1}^{N}}{\textstyle \sum_{j=1}^{M}} \Big(\alpha E_{ij}^L\big((1-x_{ij}^T),y_{ij}^{T1},y_{ij}^{T2}\big) 
\\ &\qquad\qquad\qquad\qquad+\beta E_{ij}^O\big(x_{ij}^T,y_{ij}^{T1},y_{ij}^{T2}\big)\Big)\\
&\text{s.t.}\quad  \left\{\begin{array}{lc}
{\textstyle \sum_{i=1}^{N}} x_{ij}^Ty_{ij}^{T1} \leq R_j ,  &\forall {j\in \mathcal{M}};\\
t\left(x_{i}^T,y_{i}^{T1},y_{i}^{T2}\right)\le T_{i},  &\forall {i\in \mathcal{N}};\\
0\le {\textstyle \sum_{j=1}^{M}} y_{ij}^{T2}\le 1 ,  &\forall{i\in \mathcal{N}};\\
x_{ij}^T\in \left \{0,1\right \} , &\forall{i\in \mathcal{N}}  , \forall{j\in \mathcal{M}};
\end{array}\right.
\end{split}
\end{equation}
where \(E_{ij}^L(\cdot)\) and \(E_{ij}^O(\cdot)\) represent the energy consumption functions for executing the task locally and on edge server, respectively.
\(t(\cdot)\) is the time function for task \(i\).
\(\alpha\) and \(\beta\) are energy consumption weights for local processing and offloaded processing, respectively.
In the constraints, \(R_j\) represents the total available computational resource at edge server \(j\), and \(T_i\) is the time limitation for task \(i\).

\textbf{Power Control:}
Power control problems deal with managing the power of transmitters (BSs for downlink or user devices for uplink) to optimize network quality-of-service. 
Here, variables \(x_{ij}\) and \(y_{ij}\) in problem (\ref{P1}) can be respectively mapped as BS association \(x_{ij}^{P}\) and power allocation \(y_{ij}^{P}\) in these problems.
The objective function \(f(\boldsymbol{x,y})\) could be interference, energy consumption, signal quality, network capacity, etc. 
Taking the example of minimizing energy consumption, \(f_P\left(\boldsymbol{x}^P,\boldsymbol{y}^P\right)={\textstyle \sum_{i=1}^{N}}{\textstyle \sum_{j=1}^{M}} y_{ij}^Px_{ij}^P\), and the problem can be formulated as follows: 
\begin{equation}
\begin{split}
&\min_{\boldsymbol{x}^P,\boldsymbol{y}^P} \,\,\, \ {\textstyle \sum_{i=1}^{N}}{\textstyle \sum_{j=1}^{M}} y_{ij}^Px_{ij}^P\\
&\text{s.t.}\quad  \left\{\begin{array}{lc}
\sum_{j=1}^{M} R_{ij}\left(y_{ij}^P\right) \ge D_{i} ,  &\forall{i\in \mathcal{N}}  ;\\
\sum_{j=1}^{M} x_{ij}^P=1 ,  &\forall{i\in \mathcal{N}};\\
\sum_{j=1}^{M} y_{ij}^P \leq P_{\textrm{max}} ,  & \forall {j\in \mathcal{M}};\\
x_{ij}^P\in \left \{0,1\right \} , &\forall{i\in \mathcal{N}}  , \forall{j\in \mathcal{M}};
\end{array}\right.
\end{split}
\end{equation}
where \(R\left(y_{ij}^P\right)\) is the data rate at power \(y_{ij}^P\), \(D_i\) is the minimum data rate requirement for user \(i\), and \(P_\textrm{max}\) is the maximum power available at BS \(j\).

\vspace{-5pt}
\subsection{{Challenges of Solving Problem (1) in Wireless Networks}}
The commonly formulated problem (\ref{P1}), in wireless networks, is generally intractable. 
The main difficulty lies in the binary variables, which make the objective function and constraints non-differentiable, and, thus, we cannot use a gradient-based approach to solve problem (\ref{P1}). 
In addition, the large-scale nature of wireless networks, characterized by dense devices and base stations, increases the size of the problem and leads to an exponential explosion in solution space. 
This exponential computational complexity makes it infeasible to evaluate every possible solution of the problem \cite{ref7}. 
The complexity of problem (\ref{P1}) also arises from the coupling of binary variables \(\boldsymbol{x}\) and continuous variables \(\boldsymbol{y}\), 
particularly the cross term \(\boldsymbol{xy}\). These cross terms are non-convex and commonly appear in wireless networking optimization problems. 
This interdependence means that optimizing \(\boldsymbol{x}\) while fixing \(\boldsymbol{y}\) (or vice versa) may not lead to an optimal solution, 
as the optimality of one set of variables directly affect the other. 
Furthermore, the diverse QoS requirements complicate the constraints, leading to the non-convexity of the problem caused by feasible region. 
The non-convexity of objective function and constraints makes it challenging for optimization to avoid local optima traps and ensure that the solutions are as close to the global optimum as possible. 


\section{A Learning-based Optimization Approach for Solving 0-1 Mixed Optimization Problems}
Our approach combines optimization theory with RL, using relaxed solutions to provide information for RL to efficiently deal with binary variables. 
In this section, we consider the case in which \(f(\boldsymbol{x,y})\) and \(g_l(\boldsymbol{x,y})_{l\in \mathcal{L}}\) are convex functions\footnote{We use the term ``convex" to refer to both convex and concave functions. By convention, we use ``convex" to encompass both types of functions, allowing us to streamline the discussion and apply convex analysis techniques uniformly.}, and the non-convexity of optimization problem is caused by the binary variables \(\boldsymbol{x}\). 
\textcolor{black}{We begin with this convexity assumption to first establish our core methodology in a theoretically tractable setting. This allows for a clear presentation of how relaxed solutions can guide the RL search policy. 
However, we acknowledge that many practical problems feature non-convexity beyond the binary variables. For example, objective functions for throughput maximization often contain non-convex terms coupling user association and resource allocation. Recognizing this challenge, we dedicate Section \uppercase\expandafter{\romannumeral 4} to extending our framework to realistic non-convex scenarios, where we employ techniques like convex envelopes and constraint transformation to adapt our approach.}

In our proposed approach, the binary variables are first relaxed to continuous variables, then we can easily find the optimal relaxed solution using convex optimization methods due to the convexity of functions \(f(\boldsymbol{x,y})\) and \(g_l(\boldsymbol{x,y})\).
When an episode is completed, the binary solution is then obtained through RL. Given the binary solution, an optimization-based method processes the continuous variables and obtains the objective function value. 
This value is then fed back into RL to update the policy for obtaining binary solutions.
Repeat until convergence or the termination condition is met.


\vspace{-10pt}
\subsection{Transformation to Sequential Decision Problems}

\begin{figure*}[t]
    \centering
    \includegraphics [scale=0.65,angle=0]{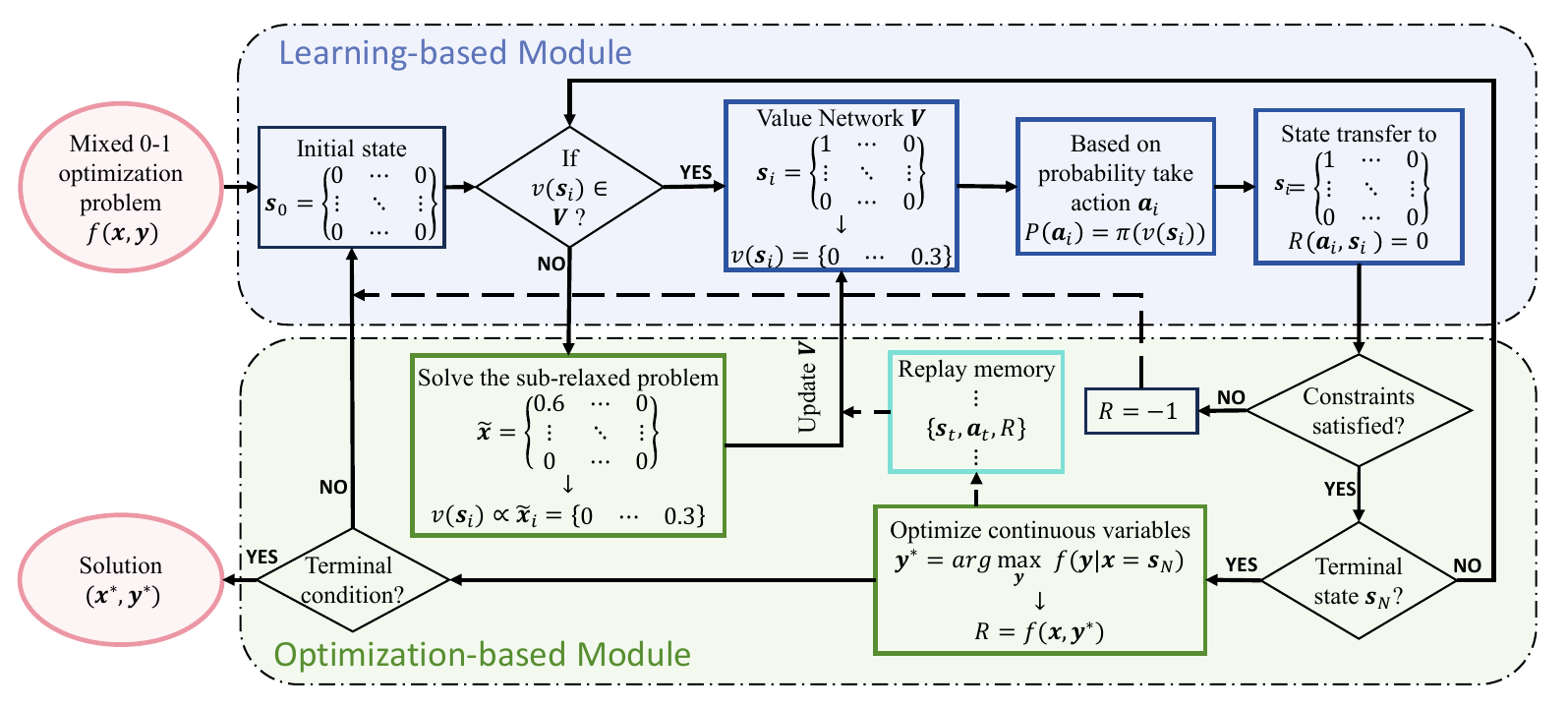}
    \caption{Flowchart of the proposed approach.}
    \label{proposed_approach}
\end{figure*}

The process of optimizing the binary variables in problem (\ref{P1}) can be equivalently represented as an MDP. 
By doing so, the optimal policy of the MDP would correspond to the optimal solution of the optimization problem.
The components of the transformed MDP framework are described next. 

\textit{State Space:} 
The state \(\boldsymbol{s}_i=\{0,1\}^{N\times M}\) represents the binary variable matrix \(\boldsymbol{x}\), where the values of first \(i\) rows have already been determined. 
Construct a 0-matrix \(\boldsymbol{s}_0=\{0\}^{N\times M}\) with the same dimension of the binary variables as the \textit{initial state}. 
For a problem with a variable size of \(N\times M\), \(N\) decisions are needed to determine which element of each row of \(\boldsymbol{x}\) is equal to \(1\), and each decision has \(M\) choices. 
Therefore, there are a total of \(M^N\) \textit{terminal states} \(\boldsymbol{s}_N\), i.e., the size of possible solution space, and there are \({\textstyle \sum_{i=1}^{N-1}} M^{i} \) \textit{intermediate states}.

\textit{Action Space:} 
There are \(M\) possible actions \(\{\boldsymbol{a}_i^1\cdots\boldsymbol{a}_i^M\}\) for each intermediate state \(\boldsymbol{s}_i\) representing an incomplete solution, where \(\boldsymbol{a}_i^k\in\{0,1\}^{M}\). 
The current state is extended by updating the next decision row as \(\boldsymbol{s}_i[i,:]= \boldsymbol{a}_i\).

\textit{{State Transition:}} 
The state transition is determined by the action we take. 
{Given a intermediate state} \(\boldsymbol{s}_i\), the system transitions to state \(\boldsymbol{s}_{i+1}\), {as described by:} 
\begin{equation}\label{p[s'|s,a]}
    P\left[\boldsymbol{s}_{i+1}|\boldsymbol{s}_i\right]=\pi(\boldsymbol{a}_i), 
\end{equation}
where \(\pi(\boldsymbol{a}_i)\) represents the probability that the search strategy takes action \(\boldsymbol{a}_i\). 

\textit{Reward Function: }The reward function can be defined as \(R(\boldsymbol{s}_{i},\boldsymbol a_i,\boldsymbol{s}_{i+1})=0\) for all transitions until a solution is reached. 
When a complete solution is achieved, \(R(\boldsymbol{s_{i}},\boldsymbol a_i,\boldsymbol{s}_{i+1})\) can be the objective function value of the 0-1 optimization problem.
\begin{equation}\label{R(s)}
    R(\boldsymbol{s}_{i},\boldsymbol{a}_i,\boldsymbol{s}_{i+1})=
\begin{cases}
 0,             & i<N;
 \\f\left(\boldsymbol{x}={\boldsymbol s}_N,\boldsymbol{y}'\right),  & i=N;
\end{cases}
\end{equation}
where  \(\boldsymbol{y}'=\arg\max_{\boldsymbol{y}} f(\boldsymbol{y}|\boldsymbol{x}=\boldsymbol{s}_N)\).

Given these mappings, it is evident that any terminal state of the MDP corresponds to a feasible solution to the binary part of a 0-1 mixed optimization problem. 
The process of obtain binary solution is shown as follow: 
\begin{equation*}
    \boldsymbol{s}_0{\overset{\boldsymbol {a}_0}{\rightarrow}}\boldsymbol{s}_1{\overset{\boldsymbol {a}_1}{\rightarrow}}\boldsymbol{s}_2 \cdots {\overset{\boldsymbol {a}_{N-1}}{\rightarrow}}\boldsymbol{s}_N=\boldsymbol{x}.
\end{equation*}
Due to the large state space and delayed rewards, RL is employed to solve the MDP. 
Once \textit{terminal state} \(\boldsymbol{s_N}\) is reached, determining the suboptimal allocation of continuous variables becomes straightforward, i.e.,  \(\boldsymbol{y}'=\arg\max_{\boldsymbol{y}} f(\boldsymbol{y}|\boldsymbol{x}=\boldsymbol{s}_N)\), s.t. constraints (1a). 
However, exploring all possibilities is not feasible due to the huge state space. 
In order to improve the efficiency of exploration, relaxed solutions are utilized as prior information to guide the RL exploration. 
\subsection{Relaxed Solutions as Priors for RL }
Given that \(f(\boldsymbol{x,y})\) and \(g_l(\boldsymbol{x,y}),\, \forall l \in \mathcal{L}\), in problem (\ref{P1}) are convex functions, 
the problem becomes a convex optimization problem after relaxing \(x_{ij} \in \{ 0,1 \}\) to \(\widetilde x_{ij} \in [0,1]\). 
\begin{align} \label{P2}
\quad \max_{\boldsymbol{{\widetilde x},y}}\quad & f(\boldsymbol{{\widetilde x},y})\\
\text{s.t.}\quad  &g_l(\boldsymbol{{\widetilde x},y}) \leq 0,\quad \forall{l\in \mathcal{L}};\tag{\ref{P2}{a}}\\
&\sum_{j=1}^{M}{\widetilde x}_{ij}=1,\quad \forall{i\in \mathcal{N}};\tag{\ref{P2}{b}}\\
&{\widetilde x}_{ij}\in [0,1],\quad \forall{i\in \mathcal{N}}, \forall{j\in \mathcal{M}}.\tag{\ref{P2}{c}}
\end{align}
The relaxed solutions to problem (\ref{P2}) can be easily obtained using convex optimization methods.
Then, the relaxed solutions are used to guide RL in prioritizing the exploration of regions with a higher likelihood of containing better solutions to the original problem. 
To formalize the concept of such regions in the solution space, we define the following term: 
\begin{definition}
A \emph{high potential zone (HPZ)}, is a subset of the possible solution space \(X_d \subseteq X\), has higher likelihood of containing suboptimal solutions compared to the feasible space: 
\begin{equation}
    \mathbb{E}\left[f(\boldsymbol{x}|\boldsymbol{y})|\boldsymbol{x}\in X_d\right]\ge \mathbb{E}[f(\boldsymbol{x},\boldsymbol{y})|\boldsymbol{x}\in X].
\end{equation}
\end{definition}



First, we present the following lemma \cite{Cantelli_inequality}, known as Cantelli's inequality, which provides a probabilistic bound on random variables. 

\begin{lemma} (from \cite{Cantelli_inequality})\label{lemma1}
    If \(Z\) is a random variable with mean \(\mu\) and variance \(\sigma^2\), then \(\forall a>0\),
    \begin{equation}\label{L1}
        P[Z\leq \mu -a]\leq {\frac{\sigma^2}{\sigma^2+a^2}}.
    \end{equation}
\end{lemma}

Then, we introduce the following proposition, to characterize the relationship between the relaxed solution and the HPZ in the context of a 0-1 mixed optimization problem.

\begin{proposition}\label{Propo1}
    For a 0-1 mixed optimization problem
\begin{equation*}
    \max\limits_{\boldsymbol{x}\in X} f(\boldsymbol{x,y}),\,X=\{0,\,1\}^{N\times M},
\end{equation*}
if the relaxed problem
\begin{equation*}
    \max\limits_{\boldsymbol{\widetilde x}\in \widetilde X} f(\boldsymbol{\widetilde x, y}), \widetilde X=[0,\,1]^{N\times M}
\end{equation*}
is convex, the neighborhood of the relaxed optimal solution \(U({\widetilde{\boldsymbol{x}} }^* ,d)\) is the HPZ of original problem, where \(d=\mathbb{E}\left[||\boldsymbol{x}-{\widetilde{\boldsymbol{x}} }^*|| , {\boldsymbol{x}}\in U\right]\).
\end{proposition}

\begin{proof}

Assume there are two subsets \(\mathcal{O}_1 \left({\widetilde{\boldsymbol{x}} }^* ,d_1\right)\) and \(\mathcal{O}_2\left({\widetilde{\boldsymbol{x}} }^* ,d_2\right)\), where \(d_1\leq d_2\), each contains at least 1 feasible binary solution.
Let \({\widetilde \mu}_1\) and \({\widetilde \sigma}_1\) represent the mean and standard deviation of \(f\left(\boldsymbol{\widetilde x, y}\right)|\boldsymbol{\widetilde{\boldsymbol{x}} }\in \mathcal{O}_1\left({\widetilde{\boldsymbol{x}} }^* ,d_1\right)\), and let \({\widetilde \mu_2}\) and \({\widetilde \sigma_2}\) represent the mean and standard deviation of \(f\left(\boldsymbol{\widetilde {x}} , \boldsymbol{y}\right)|\boldsymbol{\widetilde x}\in \mathcal{O}_2\left({\widetilde{\boldsymbol{x}} }^* ,d_2\right)\).

It is obvious that  
\({\widetilde \mu}_1 \ge {\widetilde \mu_2}\) and \({{\widetilde \sigma}_1}\leq {{\widetilde \sigma_2}}\), 
thus 
\begin{equation}
    \label{p1_1}
    {\widetilde \mu}_1-k{\widetilde \sigma}_1 \ge {\widetilde \mu_2}-k{\widetilde \sigma_2},\quad k\ge 0.
\end{equation}

According to Lemma \ref{lemma1}, for any distribution random variable \(Z\) with \(\mu\) and \(\sigma\), let \(a=k\sigma\), then we have
\begin{align}
    P\left[Z\ge \mu -k\sigma\right]&\ge 1-P\left[Z\leq \mu -k\sigma\right]\label{L1-1}
    \\&\ge {\frac{k^2}{1+k^2}}.\nonumber
\end{align}

According to (\ref{L1-1}), there is 
\begin{align}
    P\left[f(\boldsymbol{\widetilde x,y})\ge{\widetilde \mu}_1-k_1{\widetilde \sigma}_1|\widetilde {\boldsymbol{x}}\in \mathcal{O}_1\right] \ge k_1^2/1+k_1^2,\\
    P\left[f(\boldsymbol{\widetilde x,y})\ge{\widetilde \mu}_2-k_2{\widetilde \sigma}_2|\widetilde {\boldsymbol{x}}\in \mathcal{O}_2\right] \ge k_2^2/1+k_2^2.
\end{align}

For \(\mathcal{O}_1\), consider independent functions \(f_1,f_2,...,f_n\) with the same relaxation value distribution \(\widetilde \mu_1\) and \(\widetilde \sigma_1\), 
denote radom variables \(F_1^{(1)},F_1^{(2)},\dots,F_1^{(n)}\) as the objective values from these functions, with \(\mu^{(1)}_1,\mu^{(2)}_1,\dots,\mu^{(n)}_1\) and \(\sigma^{(1)}_1,\sigma^{(2)}_1,\dots,\sigma^{(n)}_1\). Here is
\begin{equation}
    P\left[\mu^{(i)}_1\ge (\widetilde\mu_1-k\widetilde\sigma_1)\right]\ge k^2/1+k^2.
\end{equation}
When the sample problems capacity \(n\) is large enough:
\begin{equation}
    \lim_{n\to \infty}P\left[ \left|\overline{{\mu}^n_1}-{1\over n}\sum_{i=1}^{n} \mathbb{E}(\mu^{(i)}_1) \right| >\varepsilon\right]=0,
\end{equation}
where \(\overline{{\mu}^n_1}={1\over n}\left(\mu^{(1)}_1+\mu^{(2)}_1+\dots+\mu^{(n)}_1\right)\) and \(\varepsilon>0\).
\\Thus, the mean expectation of discrete objective values in \(\mathcal{O}_1\) has lower confidence bound: 
\begin{equation}
    \lim_{n\to \infty}\overline{{\mu}^n_1}{\overset{p}{\longrightarrow}}{1\over n}{\textstyle \sum_{i=1}^{n}} \mathbb{E}\left[\mu^{(i)}_1\right]\ge {\left(\widetilde\mu_1-k\widetilde\sigma_1\right)k^2\over{1+k^2}}.
\end{equation}
The same with \(\mathcal{O}_2\): 
\begin{equation}
    \lim_{n\to \infty}\overline{{\mu}^n_2}{\overset{p}{\longrightarrow}}{1\over n}{\textstyle \sum_{i=1}^{n}} \mathbb{E}\left[\mu^{(i)}_2\right]\ge {\left(\widetilde\mu_2-k\widetilde\sigma_2\right)k^2\over{1+k^2}}.
\end{equation}
There exist \(k_1^*\ge0\) and \(k_2^*\ge0\) to let these two lower bounds achieve their maximum values respectively, and with (\ref{p1_1}) there is:
\begin{equation}
    {\left(\widetilde\mu_1-k_1^*\widetilde\sigma_1\right){k_1^*}^2\over{1+{k_1^*}^2}} \ge {(\widetilde\mu_1-k_2^*\widetilde\sigma_1){k_2^*}^2\over{1+{k_2^*}^2}} \ge {(\widetilde\mu_2-k_2^*\widetilde\sigma_2){k_2^*}^2\over{1+{k_2^*}^2}}.
\end{equation}
Which means \(f(\boldsymbol{x,y})|\boldsymbol{x}\in \mathcal{O}_1\) has a better lower confidence bound of potential solution than \(f(\boldsymbol{x,y})|\boldsymbol{x}\in \mathcal{O}_2\).
\\Therefore, the neighborhood of relaxed optimal solution \(U({\widetilde{\boldsymbol{x}}}^*,d)\) is the HPZ for a 0-1 mixed optimization problem.
\end{proof}

According to the proof of Proposition \ref{Propo1}, we have: 
\begin{equation}
    \mathbb{E}\big[f(\boldsymbol{x,y})|x\in U({\widetilde{\boldsymbol{x}}}^* , d_{i})\big] \ge \mathbb{E}\big[f(\boldsymbol{x,y})|x\in U({\widetilde{\boldsymbol{x}} }^* , d_{i+1})\big],
\end{equation}
where \(U\left({\widetilde x}^* , d_{i}\right)\subseteq U\left({\widetilde x}^* , d_{i+1}\right)\).
Thus, it is justified to utilize the relaxed solutions as prior information for RL to guide the design of the search policy. 
This search policy prioritizes exploring potential solutions within the HPZ. 
In an HPZ, our proposed approach provides a probabilistic lower bound on solution quality. 
Specifically, the approach guarantees a lower confidence bound on the expected objective value of the solutions. 
An HPZ with a better lower confidence bound is given more opportunity to be explored. 
Specifically, for two subsets \(U({\widetilde x}^* , d_{i})\) and \(U({\widetilde x}^* , d_{i+1})\) the search policy assigns higher exploration probability density to the HPZ with a better confidence bound, as expressed by the following equation: 
\begin{equation}
    \pi\left[\boldsymbol{x}| \boldsymbol{x}\in U({\widetilde{\boldsymbol{x}}}^* , d_{i}) \right]\ge \pi\left[\boldsymbol{x}|\boldsymbol{x}\in U({\widetilde{\boldsymbol{x}}}^* , d_{i+1}) \right].
\end{equation}
This ensures that the search process focuses on regions with a higher likelihood of containing suboptimal solutions, thereby improving the efficiency of RL exploration. 

\textcolor{black}{
The practical implementation of this guided search policy is achieved by using the specific values from the relaxed solution to initialize the state-value function for the RL agent. 
Specifically, at any decision step, the algorithm solves a relaxed sub-problem to obtain a vector of continuous values, corresponding to each possible action. 
These values are used to calculate the action-selection probabilities in our exploration strategy. 
This method creates a soft probabilistic neighborhood by imposing a probability distribution over the action space, rather than defining hard boundaries. 
It effectively biases the search towards the promising regions indicated by the HPZ while still allowing for the exploration of less likely paths.
}


\vspace{-5pt}
\subsection{Exploration Strategy in RL}
Given the relaxed solutions, we use them as priors to guide the RL agent in exploring the HPZs in the solution space. 
In each iteration of the RL process, when a terminal state \(\boldsymbol{s}_N\) is reached, we compute the value of the optimization objective, and use this value as reward: 
 \begin{equation}
 R= \max_{\boldsymbol{y}} f(\boldsymbol{y}|\boldsymbol{x}=\boldsymbol{s}_N).
 \end{equation}
This reward is then propagated back through the previous states to update the search policy, ensuring that the agent learns which areas of the solution space yield better results. 

From (\ref{p[s'|s,a]}), we can simplify the transition probability as: 
\begin{equation}\label{p[s']}
        P(\boldsymbol{s}_{i+1}|\boldsymbol{s}_{i})=P(\boldsymbol{a}_{i}|\boldsymbol{s}_{i}).
\end{equation}
Thus, in our RL process, according to (\ref{R(s)}) and (\ref{p[s']}), the action-value function can be written as: 
\begin{equation}
\begin{split}
    q(\boldsymbol{s}_i,\boldsymbol{a}_i)&=R_{\boldsymbol{s}_i}^{\boldsymbol{a}_i}+\gamma \sum_{\boldsymbol{s}_{i+1}} P\left[\boldsymbol{s}_{i+1}|\boldsymbol{s}_i,\boldsymbol{a}_i\right]v(\boldsymbol{s}_{i+1})\\
    &=\gamma v(\boldsymbol{s}_{i+1}),
\end{split}
\end{equation}
where \(\gamma\) is discount factor and \(v(\boldsymbol{s}_i)\) is the value of state \(\boldsymbol{s}_i\).
Given this, the policy for selecting actions will be given by:
\textcolor{black}{
\begin{equation}\label{sele-action}
    \pi(\boldsymbol {a}_i|\boldsymbol{s}_i)=P(\boldsymbol{a}_i|\boldsymbol{s}_i)=\frac{e^{v(\boldsymbol{s}_{i+1})}}{\sum_{\boldsymbol{s}_{i+1}} e^{v(\boldsymbol{s}_{i+1})}}.
\end{equation}
}
As the agent interacts within the solution space, its policy is updated based on the objective values it receives. 
Next, we state the following proposition to establish the convergence of the policy.

\begin{proposition}
Suppose there exist feasible solutions in the feasible domain and the objective value is bounded. 
Then, as the number of iterations \(k\) tends to infinity, the objective value converges, i.e.,
\begin{equation*}
    \lim_{k\rightarrow \infty} f(\boldsymbol{x}_{k+1},\boldsymbol{y}'_{k+1})-f(\boldsymbol{x}_{k},\boldsymbol{y}'_{k})\leq \epsilon, \forall{\epsilon\ge0}.
\end{equation*}
\end{proposition}
\begin{proof}
Let \(\pi_{k+1}\) be the updated policy after an iteration, and \(\pi_k\) represent the previous policy. 
For states that do not satisfy the constraints, the probability of selecting actions that lead to these states decreases after the policy update:
\begin{equation}\label{P[sT]}
    P_{\pi_{k+1}}\big[\boldsymbol{a}_i|\boldsymbol{s}_i,\boldsymbol{s}_{i+1}\in S_T\big] \leq P_{\pi_k}\big[\boldsymbol{a}_i|\boldsymbol{s}_i,\boldsymbol{s}_{i+1}\in S_T\big],
\end{equation}
where \(S_T\) is the set of states which are not satisfied with constraints, and \(V(\boldsymbol{s}_i|\boldsymbol{s}_{i}\in S_T)<0\). 
For states where the constraints are satisfied, the probability of selecting suboptimal actions increases:
\begin{equation}\label{P[sS]}
    P_{\pi_{k+1}}\big[\boldsymbol{a}_i^*|\boldsymbol{s}_i,\boldsymbol{s}_{i+1}\notin S_T\big]\geq P_{\pi_k}\big[\boldsymbol{a}_i^*|\boldsymbol{s}_i,\boldsymbol{s}_{i+1}\notin S_T\big], 
\end{equation}
where \(\boldsymbol{a}_i^*=\arg\max_{\boldsymbol{a}_i} q(\boldsymbol{s}_i,\boldsymbol{a}_i)\). 

The value of state \(\boldsymbol{s}_i\) under policy \(\pi\) is:
\begin{equation}
\begin{aligned}\label{V(s)}
     V_{\pi}(\boldsymbol{s})&=\sum_{\boldsymbol{a}_i} \pi(\boldsymbol {a}_i|\boldsymbol{s}_i)\big[R(\boldsymbol{s}_{i+1})+\gamma V_{\pi}(\boldsymbol{s}_{i+1})\big].
\end{aligned}
\end{equation}
From (\ref{p[s']}), it can be rewrite as: 
\begin{equation}
\begin{aligned}\label{V(s)_}
     V_{\pi}(\boldsymbol{s}_i) &=\sum_{\boldsymbol{a}_i}P_{\pi}\big[\boldsymbol{a}_i|\boldsymbol{s}_i,\boldsymbol{s}_{i+1}\in S_T\big]\big[R(\boldsymbol{s}_i)+\gamma V_{\pi}(\boldsymbol{s}_{i+1})\big]
     \\&+\sum_{\boldsymbol{a}_i}P_{\pi}\big[\boldsymbol{a}_i|\boldsymbol{s}_i,\boldsymbol{s}_{i+1}\notin S_T\big]\big[R(\boldsymbol{s}_i)+\gamma V_{\pi}(\boldsymbol{s}_{i+1})\big].
\end{aligned}
\end{equation}
Substitute (\ref{P[sT]}), and (\ref{P[sS]}) into (\ref{V(s)_}), we have: 
\begin{equation}
     V_{\pi_{k+1}}(\boldsymbol{s})\geq V_{\pi_k}(\boldsymbol{s}), 
\end{equation}
which means the value function is monotonically increasing. 
Therefore, the objective value is non-decreasing as the number of iterations increases: 
\begin{equation}
     f(\boldsymbol{x}_{k+1},\boldsymbol{y}'_{k+1})\ge f(\boldsymbol{x}_{k},\boldsymbol{y}'_{k}).
\end{equation}
Since the objective value is bounded, then we have:
\begin{equation*}
    \lim_{k\rightarrow \infty} f(\boldsymbol{x}_{k+1},\boldsymbol{y}'_{k+1})-f(\boldsymbol{x}_{k},\boldsymbol{y}'_{k})\leq \epsilon, \forall{\epsilon\ge0}.
\end{equation*}
\end{proof}
\vspace{-10pt}
Proposition 2 indicates that by updating the searching policy based on the rewards obtained from the objective values, the agent will eventually converge. 
The exploration strategy presented in this section allows the RL agent to efficiently explore the HPZs in the solution space. 
The framework of the proposed approach is illustrated in Fig.~\ref{proposed_approach} and is summarized in Algorithm \ref{alg:1}.
\textcolor{black}{
If an action at any step leads to a state that violates a problem constraint, the episode terminates prematurely, and a penalty reward of \(R=-1\) is assigned. This teaches the agent to avoid such paths. For all valid, non-terminal steps within an ongoing episode, the immediate reward is \(R=0\). This ensures the agent focuses on the final outcome rather than the path length.
}

\begin{algorithm}[t]
    \caption{RL-based Optimization Algorithm for 0-1 Mixed Problem}
    \label{alg:1}
    \begin{algorithmic}[1]
    \STATE  Initialize states value \(V(\boldsymbol{s})\); initialize replay buffer \(D\) with capacity \(d\); and let \(m = 0\).
    
    \STATE \textbf{Repeat}(for each episode):
    \STATE \quad Initialize \(\boldsymbol{s}\leftarrow \{0\}^{N\times M},\,i=0\).
    \STATE \quad \textbf{Repeat}(for each step of an episode):
    \STATE \qquad  If \(\boldsymbol{s}\notin \textrm{dom}(V(\boldsymbol{s}))\):
    \STATE \qquad\quad  Solve relaxed sub-problem: 
         \\\qquad\qquad \(\boldsymbol{\widetilde x}=\underset{\boldsymbol{\widetilde x}\in \boldsymbol{\widetilde X}}{\arg\max}  f(\boldsymbol{y},\widetilde{\boldsymbol{x}}|\boldsymbol{s}).\)
    \STATE \qquad\quad  Initialize \(V(\boldsymbol{s})=\boldsymbol{\widetilde x}\).
    
    \STATE \qquad For a given state \(\boldsymbol{s}\), generate an action \(\boldsymbol a\leftarrow\pi_v(\boldsymbol a|\boldsymbol{s})\),
         \\\qquad using policy derived form \(V(\boldsymbol{s})\).
    \STATE \qquad Update \(\boldsymbol{s}\leftarrow \boldsymbol{s}+\boldsymbol a\),\(i=i+1\).
    \STATE \quad \textbf{Until:} \(\boldsymbol{s}\) is a terminal state or constraints cannot be satisfied.
    \STATE \quad Reward 
                \begin{equation*}
                    R=
                    \begin{cases}
                     -1,      &\boldsymbol{s}\text{ is not a terminal state};
                     \\\max_y f(\boldsymbol{y}|\boldsymbol{x=s}),  & \boldsymbol{s}\text{ is a terminal state}.
                    \end{cases}
                \end{equation*}
    \STATE \quad Update replay buffer D and states value 
    \(V(\boldsymbol{s}_i)\leftarrow \{ V(\boldsymbol{s}_i),\,R\}_{i \in \mathcal{N}}.\)
    \STATE \quad \(m = m+1\).
    \STATE \textbf{Until:} converges within a prescribed accuracy or a maximum number of iterations has reached.
    \end{algorithmic}
\end{algorithm}

The complexity of the approach can be analyzed by considering both the computation of HPZs and the iterative RL process. 
First, obtaining the HPZs involves solving a relaxed optimization problem, which is a convex optimization problem. 
The complexity of solving this convex relaxation, assuming the use of standard convex optimization algorithms, is typically \(O\big((N\times M)^3\big)\) \cite{boyd2004convex}. 
Then, we define the number of states as \(|S|\), 
the number of actions as \(|A|\), the complexity of a single RL policy update is \(O\left(|S|\cdot|A|\right)\).
Over \(E\) iterations, the overall complexity becomes \(O\big(T\cdot(N\times M)^3+T\cdot|S|\cdot|A|\big)\). 
Our approach offers a potential improvement by using relaxed solutions to focus exploration, which can reduce the search space compared to blind exploration. 
By guiding the agent toward HPZs, we expect a faster convergence to suboptimal solutions. 
This makes our approach computationally more efficient in cases where the solution space is large, compared to traditional RL methods that explore the entire solution space more uniformly.





{
  \widowpenalty=0
  \clubpenalty=0

\section{Extensions to Non-convex Scenarios}
In wireless network communication scenarios, there are many cases in which the objective function and/or constraints are non-convex.
For example, optimizing the throughput via power control in a multi-user environment often leads to a non-convex objective function because of the interference between transmitters.
the relationship between transmit power and interference is non-linear. 
Additionally, constraints related to QoS requirements can often be non-convex, making the optimization problems challenging.
In this section, we discuss how to apply our proposed optimization approach to the problems with non-convex objective function and/or constraints.

\vspace{-15pt}
\subsection{Non-Convex Objective Function Problems}
{Optimization problems such as bandwidth allocation, power control, energy efficiency maximization, and multi-carrier resource allocation typically have non-convex objective functions. 
This is mainly due to the presence of cross terms or fractional terms.} 
Let \(\hat{f}(\boldsymbol{x,y})\) be a non-convex objective function. 
The proposed solution in Section \uppercase\expandafter{\romannumeral 3} cannot be directly applied since it is challenging to obtain the optimal relaxed solutions considering the presence of multiple local optimum and/or saddle points of \(\hat{f}(\widetilde{\boldsymbol{x}},\boldsymbol{y})\). 
To extend our proposed solution to this case, one straightforward way is using the convex envelope \(\textrm{conv}(\hat{f})(\widetilde{\boldsymbol{x}},\boldsymbol{y})\) to approximate \(\hat{f}(\widetilde{\boldsymbol{x}},\boldsymbol{y})\). If the convex envelope is difficult to obtain, another efficient way is estimating \(\hat{f}(\boldsymbol{x,y})\) with numerical estimation methods.
}
The double conjugate is commonly employed to transform a non-convex function into a convex approximation. 
Formally, the double conjugate \(\hat{f}^{**}\) of function \(\hat{f}\), is defined as: 
\begin{align}
    \hat{f}^{**}(\boldsymbol{z})=\sup_{\boldsymbol{z}^*} {\{\left\langle \boldsymbol{z}^*,\boldsymbol{z} \right\rangle - \hat{f}^*(\boldsymbol{z}^*) \}},
\end{align}
where \( \hat{f}^*(\boldsymbol{z}^*)=\sup_{\boldsymbol{z}} {\{\left\langle \boldsymbol{z}^*,\boldsymbol{z} \right\rangle - \hat{f}(\boldsymbol{z}) \}}.\)
We can reverse it to upper convex envelope for maximization problem
\begin{align}
    \hat{f}_u^{**}(\boldsymbol{z})=\inf_{\boldsymbol{z}^*} {\{-\left\langle \boldsymbol{z}^*,\boldsymbol{z} \right\rangle - \hat{f}_u^*(\boldsymbol{z}^*) \}},
\end{align}
where \( \hat{f}_u^*(\boldsymbol{z}^*)=\inf_{\boldsymbol{z}} {\{-\left\langle \boldsymbol{z}^*,\boldsymbol{z} \right\rangle - \hat{f}(\boldsymbol{z}) \}}.\)
\(\hat{f}_u^{**}\) is the least upper semi-continuous upper convex function that overestimates. 

Lemma \ref{lemma_2} reveals that \(\hat{f}^{**}\) is the convex envelope of \(\hat{f}\), providing the closest convex approximation. 
Then, Proposition \ref{proposition_3} demonstrates that the double conjugate does not move the global optimum of \(\hat{f}\) while being convex and more tractable for optimization.

\begin{lemma} \label{lemma_2}
{(Fenchel–Moreau–Rockafellar Theorem \cite{biconjugate_book})} For every function \(f\), whenever \(f\) admits a continuous affine minorant, there is 
\begin{equation}
    f^{**}= \textrm{conv}(f), 
\end{equation} 
where \(\textrm{conv}(f):=\sup\{g: g \text{ is convex, semi-continuous,}\) \(\text{ and } g\leq f\}\).

The same holds true for \(f_u^{**}\).
\end{lemma}

\begin{proposition}\label{proposition_3}
    Assume function \(f\) admits a continuous affine minorant and has a unique global optimum, then
    \begin{equation}
        \max f_u^{**}=f_u^{**}(\arg\max f)=\max f.
    \end{equation}
\end{proposition}
\begin{proof}
    Assume     
    \begin{equation}
        (\arg\max f_u^{**},\max f_u^{**})\neq (\arg\max f,\max f),
    \end{equation}
    then,
    \begin{equation}\label{p3_1}
        f_u^{**}(\arg\max f) > \max f,
    \end{equation}
    or 
    \begin{equation}\label{p3_2}
        \max f_u^{**}>\max f.
    \end{equation}
    It is easy to find a convex function \(g\ge f\), 
    \begin{equation}\label{p3_3}
        \max g=g(\arg\max f)=\max f.
    \end{equation}
    Then, there exists a convex function 
    \begin{equation}\label{p3_4}
        g^*=\inf\{g,f_u^{**}\}.
    \end{equation}
    From (\ref{p3_1}-\ref{p3_4}), we have \(g^*\leq f_u^{**}\), which is contradiction to Lemma 2.\\
    Therefore, \(f_u^{**}\) and \(f\) have the same global optimum point.
\end{proof}
Proposition 3 indicates that the optimization landscape of \(\hat{f}^{**}\) retains the global optimal solutions of \(\hat{f}\), making it a faithful convex proxy. 
Thus, solving the convex problem given by \(\hat{f}^{**}\) can provide prior information for the RL.
Integrating convex estimation method, our approach could be extended to problems with non-convex objective function. 
As shown in Algorithm \ref{alg:2}, an estimation method could be integrated into the steps of obtaining relaxed solutions to provide prior information for RL.

\begin{algorithm}[t]
    \caption{RL-based Optimization Algorithm for 0-1 Mixed Problem|Non Convex Objective Function}
    \label{alg:2}
    \begin{algorithmic}[1]
    \STATE Same with Algorithm 1 steps 1-5.\setcounter{ALC@line}{5}
    \STATE \qquad\quad If \(\exists\, \textrm{conv}(\hat{f})\):
    \STATE \qquad\qquad  Solve relaxed sub-problem: 
         \\\qquad\qquad \(\boldsymbol{\widetilde x}={\arg\max}_{\boldsymbol{\widetilde x}}\, \textrm{conv}(\hat{f})(\boldsymbol{y},\widetilde{\boldsymbol{x}}|\boldsymbol{s}).\)
    \STATE \qquad\quad Else: 
         \\\qquad\qquad Estimated \(\boldsymbol{\widetilde x}\) by numerical estimation methods.
    \STATE \qquad\quad  Initialize \(V(\boldsymbol{s})=\boldsymbol{\widetilde x}\). 
    
    \STATE Same with Algorithm 1 steps 9-11.\setcounter{ALC@line}{12}
    \STATE \quad Reward 
                \begin{equation*}
                    R=
                    \begin{cases}
                     -1,      &\boldsymbol{s}\text{ is not completed};
                     \\\max_y \textrm{conv}(\hat{f})(\boldsymbol{y}|\boldsymbol{x=s}),  & \boldsymbol{s}\text{ is completed}.
                    \end{cases}
                \end{equation*}
    \STATE Same with Algorithm 1 steps 12-14.
    \end{algorithmic}
\end{algorithm}
\vspace{-10pt}
\subsection{Non-Convex Constraints Problems}
Besides the non-convexity in the objective function discussed in Section \uppercase\expandafter{\romannumeral 4}. A, let us discuss how to extend the proposed solution to problems with non-convex constraints. 
Let \(\hat{g}_l(\boldsymbol{x,y})_{l\in\mathcal{L}}\) be the non-convex constraints, which typically lead to a feasible region non-convex. 
The non-convexity of feasible region makes it harder to search for the global optimum.
Three potential ways could be introduced to extend our proposed solution to this case: convex hull \cite{convexhull}, Lagrangian dual method \cite{L_nonconvex_constrants}, and/or arc consistency judgment \cite{arc_consistency}.

The convex hull approach transforms the non-convex feasible region into a convex hull using relaxation techniques, enabling efficient convex optimization \cite{convexhull}. 
The convex hull retains optimality but compromises feasibility. 
The Lagrangian dual method handles coupling non-convex constraints by decomposing them into independent subproblems, which are easier to handle independently \cite{L_nonconvex_constrants}. 
\textcolor{black}{The arc consistency method reduces the search space by enforcing constraints between variables. It works by iteratively removing inconsistent values from the variables domains \cite{arc_consistency}. This process simplifies the problem and helps manage non-convexity.}

By applying these three methods individually or in combination, our approach can be extended to problems with non-convex constraints. 
As shown in Alg. \ref{alg:3}, the constraints, which are difficult to approximate convexly, could be ignored during solving the relaxed problem, making feasible region convex. 
Then, consider arc consistency in RL when initial state values, setting the values of the states do not satisfy arc consistency to -1. 
\textcolor{black}{
This negative assignment serves as a strong penalty for actions that violate arc consistency. The mechanism is directly linked to our policy selection rule in (\ref{sele-action}), where the probability of selecting an action is proportional to the value of the subsequent state. During the policy calculation, the negative value ensures that the probability of selecting the invalid action will tend to zero. As a result, the RL agent is steered away from exploring this infeasible path, which prunes this branch from the search space.
}

\begin{algorithm}[t]
    \caption{RL-based Optimization Algorithm for 0-1 Mixed Problem|Arc Consistency Judgment}
    \label{alg:3}
    \begin{algorithmic}[1]
    \STATE  Same with Algorithm 1 steps 1-4.\setcounter{ALC@line}{4}
    \STATE \qquad  If \(\boldsymbol{s}\notin \textrm{dom}(V(\boldsymbol{s}))\):
    \STATE \qquad\quad  Solve relaxed sub-problem: 
         \\\qquad\qquad \(\boldsymbol{\widetilde x}={\arg\max}_{\boldsymbol{\widetilde x}\in \boldsymbol{\widetilde X}}  \textrm{conv}(f)(\boldsymbol{y},\widetilde{\boldsymbol{x}}|\boldsymbol{s}).\)
    \STATE \qquad\quad  For \(j \in \mathcal{M} \):
    \STATE \qquad\qquad  If \({\boldsymbol a}_i^j \) does not satisfy arc consistency:
    \STATE \qquad\qquad\quad Let \(\widetilde{x}_{ij}=-1\).
    \STATE \qquad\quad  Initialize \(V(\boldsymbol{s})=\boldsymbol{\widetilde x}\).
    \STATE Same with Algorithm 1 steps 8-14.
    \end{algorithmic}
\end{algorithm}

\section{Simulation Results and Analysis}
In this section, we evaluate the performance of our proposed optimization approach through numerical simulations.
Traditional RL approach and B\&B approach are selected as benchmarks.
\begin{itemize}
    \item 
    \emph{B\&B approach:} This approach keeps the same process to obtain the relaxed solution first, then employing depth-first heuristic exploration to determine binary variables. This branching process continues until the binary part of solution obtained from all branches are binary.
    \item 
    \emph{Traditional RL approach:} RL is employed to explore the binary part of solutions. The initial values of the states are averaged due to the lack of prior information, which differs from the proposed approach. Arc consistency judgment is used to improve searching efficiency. After obtaining binary solution, the continuous variables are determined using convex optimization. 
\end{itemize}
For ease of comparison, the \textit{normalized objective value} representing the ratio of objective value to the upper bound of the relaxation problem is utilized. 
Note that this upper bound is not the optimal objective value. Thus it is meaningless to compare this normalized value between different figures (i.e., with different settings), but meaningful to compare it and reveal the performance gain within one figure (i.e., with the same settings). 
\textcolor{black}{
Before presenting the results, we define the key hyperparameter \(\alpha\in[0,1]\) as the learning rate. This parameter controls how much the final reward from an episode updates the state-value function, managing the trade-off between faster learning and more stable updates.}
The simulation codes are built with Python 3.8 and conducted on an Intel Core i7 CPU with 16GB of RAM.

\vspace{-10pt}
\subsection{Simulation Results for Convex Relaxed Problems}
First, we conduct the simulations for problems with a convex objective function. 
Consider the below problem as a common use case. For example, this problem can formulate a joint user-association and edge-computing resource allocation problem, as follows:
\begin{align} \label{S-P1}
\quad\max_{\boldsymbol{x,y}} \quad &\sum_i^N\sum_j^M \big[\log (1+s_{ij}x_{ij})+d_{ij}y_{ij}\big]\\
\text{s.t.}\quad  &\textstyle{\sum_{i=1}^N} y_{ij}\leq D_{j},\quad  \forall{j\in \mathcal{M}} ;\tag{\ref{S-P1}{a}}\\
&\textstyle{\sum_{i=1}^N} c_{ij}x_{ij}\leq C_{j},\quad \forall{j\in \mathcal{M}};\tag{\ref{S-P1}{b}}\\
&\textstyle{\sum_{j=1}^{M}}x_{ij}=1 , \quad \forall{i\in \mathcal{N}};\tag{\ref{S-P1}{c}}\\
&x_{ij}\in \{0,1 \} ,\quad \forall{i\in \mathcal{N}} , \forall{j\in \mathcal{M}}\tag{\ref{S-P1}{d}};
\end{align}
{where} \(x_{ij}\) is the association-decision variable, \(y_{ij}\) is the computing‑resource variable, \(s_{ij}\) and \(d_{ij}\) represent the transmission and computation utility weights, respectively; 
\(c_{ij}\) is the amount of communication resources required, and \(D_j\) and \(C_j\) represent, respectively,the computation and communication resource budgets at server \(j\).


\begin{figure}
    \centering
    \includegraphics[scale=0.54]{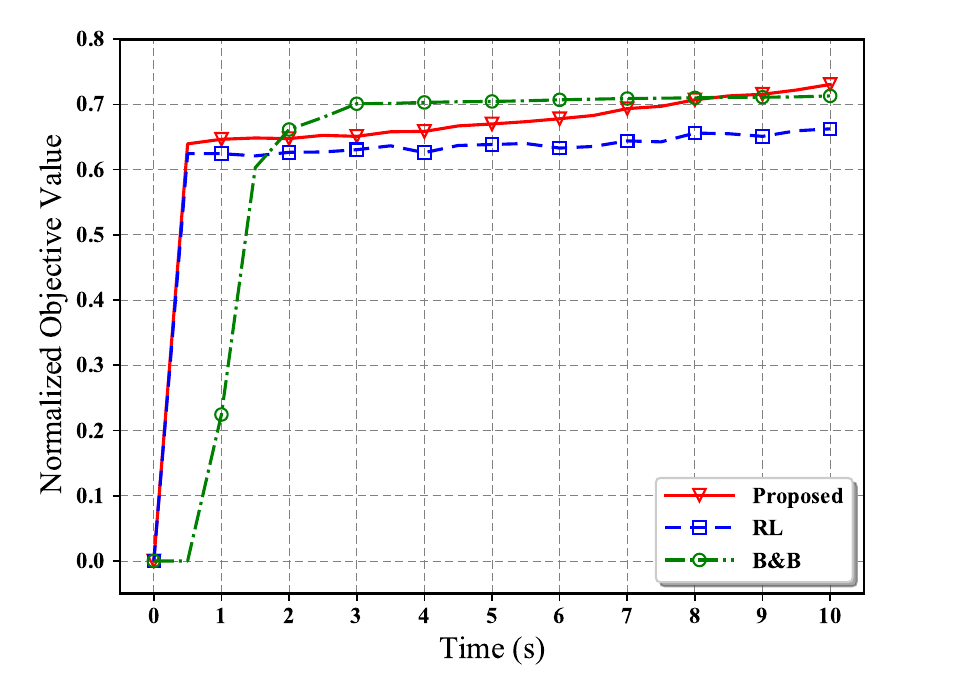}
    \caption{Running time vs. normalized objective value (\(N = 8, M = 4\)).}
    \label{c_p-tn8}
    \vspace{-1pt}
    \includegraphics[scale=0.54]{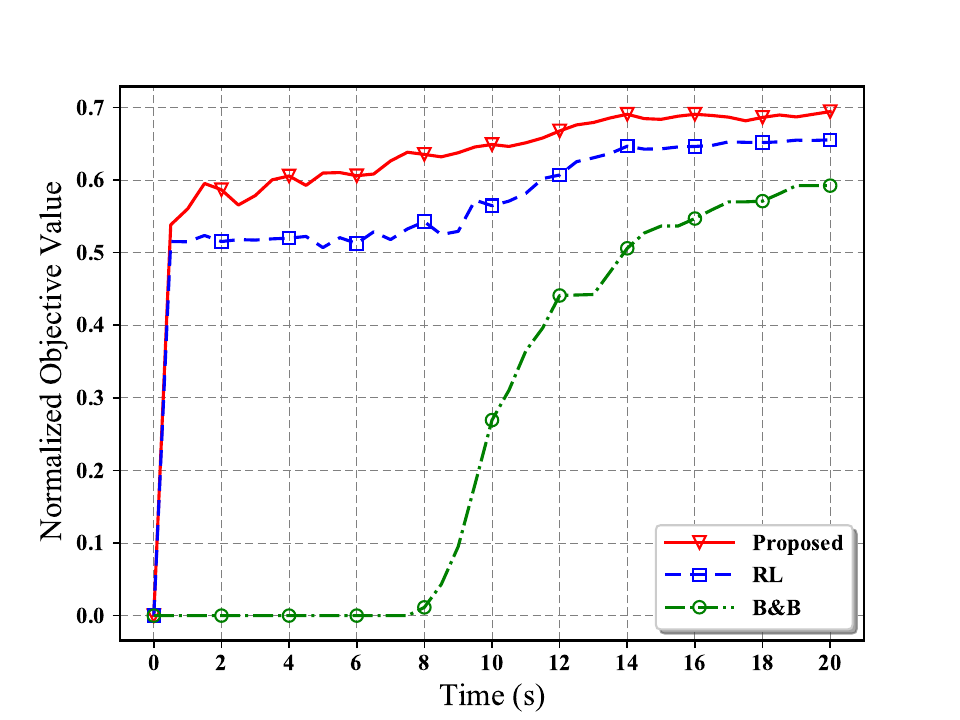}
    \caption{Running time vs. normalized objective value (\(N=20, M=5\)).}
    \label{c_p-tn20}
\end{figure}

Figs. \ref{c_p-tn8} and \ref{c_p-tn20} compare the normalized objective value over the time under two small scales for the three algorithms. From Fig. \ref{c_p-tn8} with a problem size \(N=8\), \(M=4\), it can be observed that the objective value of all three approaches increases rapidly at the beginning. 
\textcolor{black}{
For this small-scale problem, the proposed approach and the traditional RL approach yield very similar results.}
Although the B\&B approach shows a lower initial increase, its normalized objective value is slightly higher than the other two approaches. 
As shown in Fig. \ref{c_p-tn20} with a problem size of \(N=20\), \(M=5\), it becomes clear that the proposed approach outperforms the two benchmarks in terms of achieved objective value.
The proposed approach reduces the convergence time by 30\% compared to the B\&B approach, and shows a 5\% improvement in the normalized objective value compared to the RL approach, and a 15\% improvement compared to the B\&B approach.
\begin{figure}[!t]
    \centering    
    \includegraphics[scale=0.55,angle=0]{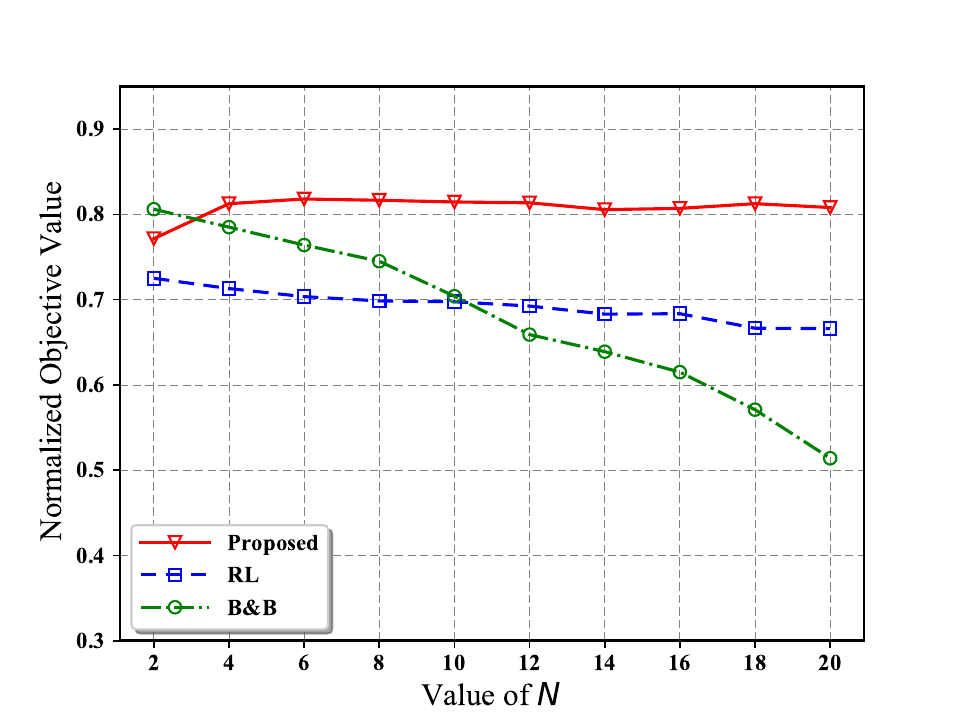}
    \caption{Comparison of normalized objective value for different values of $N$ under \(M = 5\).}
    \label{c_p-size}
\end{figure}

\begin{figure}[!t]
    \centering
    \includegraphics [scale=0.55,angle=0]{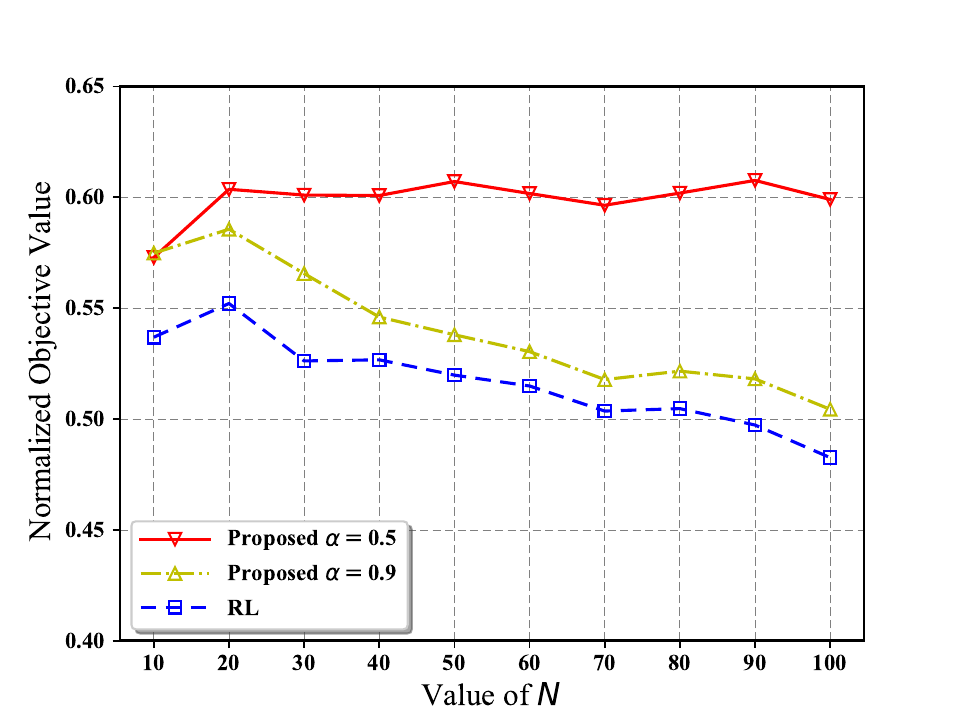}
    \caption{Comparison of normalized objective value for different values of $N$ under \(M = 10\).}
    \label{cov_p}
    \includegraphics[scale=0.55]{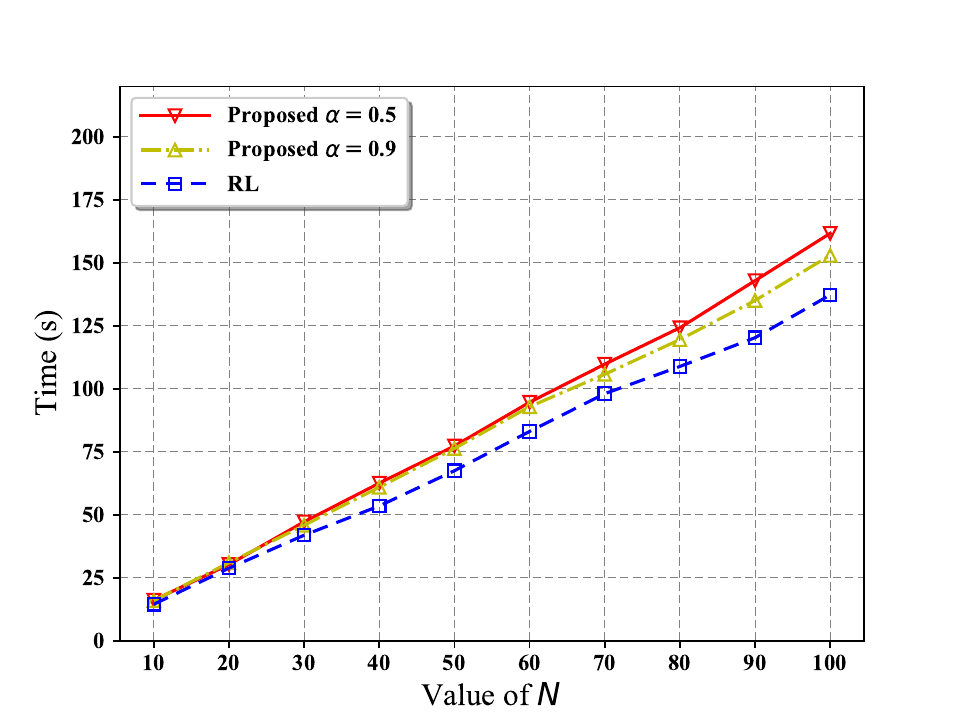}
    \caption{Comparison of running time for different values of $N$ under \(M = 10\).}
    \label{cov_t}
\end{figure}

\begin{figure}[!t]
    \centering
    \includegraphics [scale=0.55,angle=0]{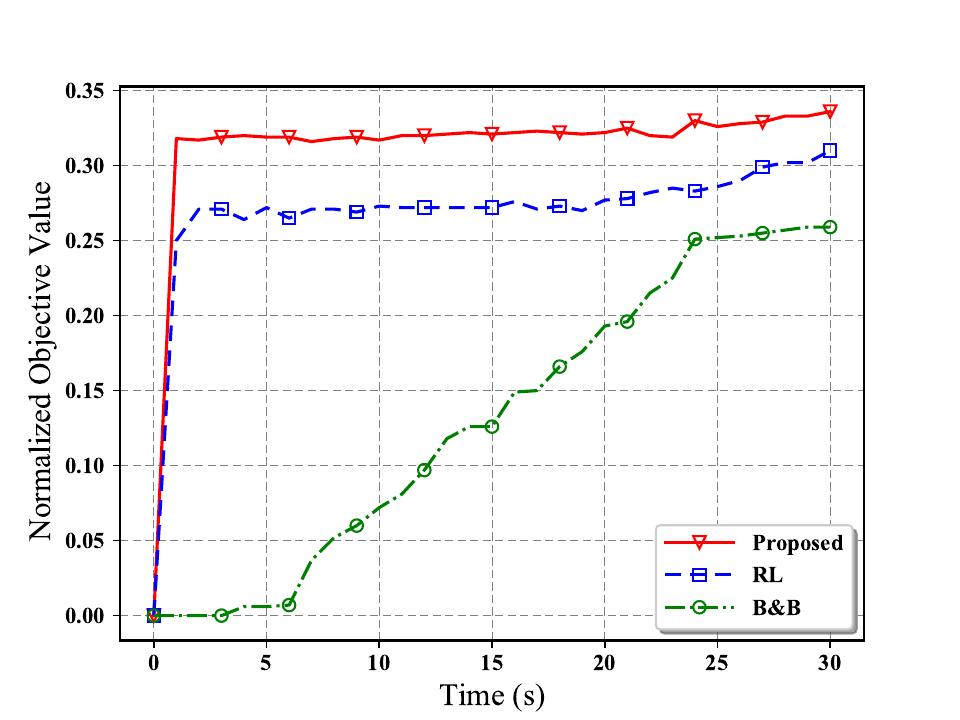}
    \caption{Running time vs. normalized objective value (\(N=10, M=5\)).}
    \label{sim_qos0.5_c}
    \centering
    \includegraphics [scale=0.55,angle=0]{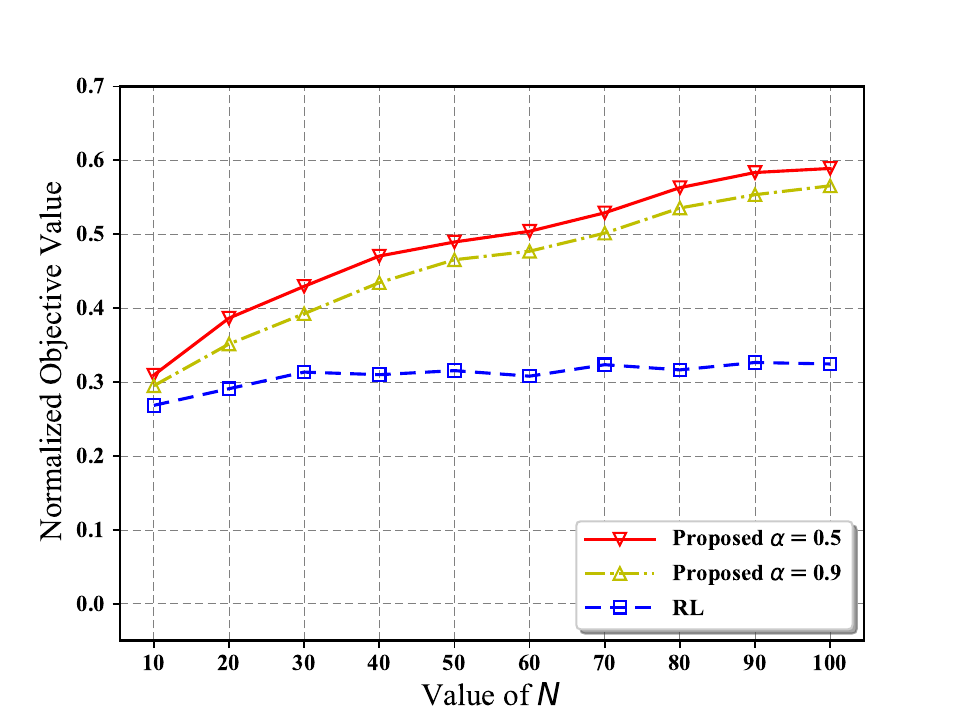}
    \caption{Comparison of normalized objective value for different values of $N$ under \(M = 10\) and \(Q = 0.1\).}
    \label{qos30_p_0.1}
    \includegraphics[scale=0.55]{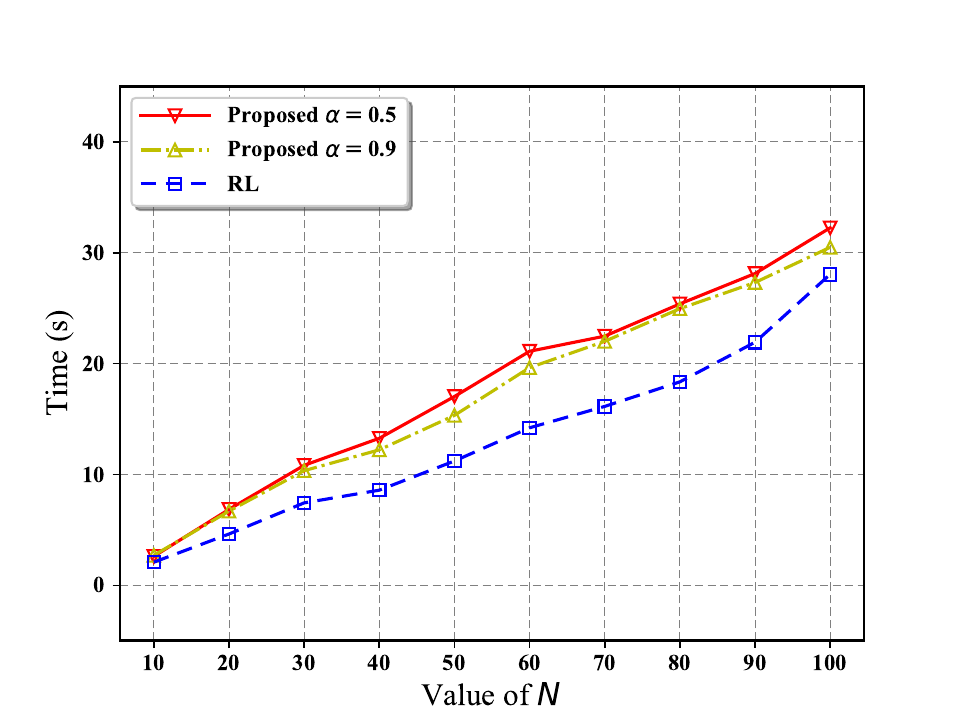}
    \caption{\textcolor{black}{Comparison of running time for different values of $N$ under \(M = 10\) and \(Q = 0.1\).}}
    \label{qos30_t_0.1}
\end{figure}
We then examine the normalized objective value for different value of \(N\) (typically \(N\) can be treated as the number of users in a network) in problem (\ref{S-P1}), as shown in Fig. \ref{c_p-size}. 
As B\&B approach cannot be deployed in large-scale problems, this simulation still focuses on small scale problem, but the next figure will examine large scale problems without B\&B approach. 
From this figure, we observe that the proposed approach always maintains the highest normalized objective value that is also stable. 
In contrast, the RL approach experiences a slight decrease in normalized objective value as the problem size increases, but it remains relatively stable overall. 
The B\&B approach demonstrates a degradation as the problem size increases, due to the high computational complexity. 

Fig. \ref{cov_p} presents the trend of normalized objective value under large scale problems. 
This figure compares the normalized objective value of the proposed approach under two different learning rate (\(\alpha=0.5\) and \(\alpha=0.9\)) with that of the RL approach. 
It is evident that the proposed approach, particularly when \(\alpha=0.5\), outperforms the RL approach by 20\% when \(N\ge70\).
Fig. \ref{cov_t} compares the running time across different problem sizes. The running time of the proposed approach is slightly higher than that of the RL approach. 
This reveals that the proposed approach achieves higher objective value with a minor compromise on computation time. The extra computing complexity mainly lies in solving relaxed subproblems. Note that computing entities in practical networks should be more powerful, which can further mitigate the extra running time.

\begin{table}[t]
\centering
\caption{Simulation Parameters}
\begin{tabular}{|l|l|}
\hline
\textbf{Parameters} & \textbf{Values}  \\ \hline
Number of BSs ($M$) & 10   \\ \hline
Bandwidth limitation of BSs  ($C$) & 20 MHz   \\ \hline
Density of BSs  & 16 \(/\textrm{km}^2\)   \\ \hline
Transmit power of UDs ($P_r$) & 20 dBm   \\ \hline
Noise power ($P_n$) & -114 dBm \cite{noise_power}  \\ \hline
Path loss model ($L(d)$) & \(L(d)=34+40\log_{10}(d)\)   \\ \hline
\end{tabular}
\end{table}

\vspace{-5pt}
\subsection{Simulation Results for Non-Convex Relaxed Problems}
Then, we conduct the simulations for non-convex relaxed problems. 
Consider the below problem as a common use case. 
For example, this problem can formulate a joint user association and bandwidth allocation, 
using a logarithmic utility function as the objective function, as follows: 

\begin{align} \label{S-P2}
\max_{\boldsymbol{x,y}}  &\sum_i^N\sum_j^M \log \big[ x_{ij}y_{ij}\log_2(1+{\gamma_{ij}}) \big] \\
\text{s.t.}\quad  &\textstyle{\sum_{j=1}^M}x_{ij}y_{ij}\log_2(1+{\gamma_{ij}})\leq Q,  \forall{i\in \mathcal{N}};\tag{\ref{S-P2}{a}}\\
&\textstyle{\sum_{i=1}^N} y_{ij}\leq C_{},\quad \forall{j\in \mathcal{M}};\tag{\ref{S-P2}{b}}\\
&\textstyle{\sum_{j=1}^{M}}x_{ij}=1 , \quad \forall{i\in \mathcal{N}};\tag{\ref{S-P2}{c}}\\
&x_{ij}\in \{0,1 \} ,\quad \forall{i\in \mathcal{N}} , \forall{j\in \mathcal{M}}\tag{\ref{S-P2}{d}};
\end{align}
where \(x_{ij}\) {is the association-decision variable}, \(y_{ij}\) {is the bandwidth resource variable}. 
\textcolor{black}{To ensure a generic and unbiased evaluation, we model a standard wireless network topology. In the simulations, BSs are placed on a uniform grid within a square area, and UDs are randomly distributed throughout the area. 
 This setup avoids any performance bias that might result from a fixed, structured topology and is representative of typical urban or suburban deployments. 
The specific parameters for path loss, transmit power, and resource constraints are then applied to this randomly generated topology and are detailed in Table \uppercase\expandafter{\romannumeral 1}.}

We start with a small size problem with \(M=5\) and \(N=10\), and evaluate the normalized objective value in problem (\ref{S-P2}) over time. 
As shown in Fig. \ref{sim_qos0.5_c}, we see that the proposed approach always achieves the highest and most stable objective value compared to the two benchmarks. 
Moreover, the proposed approach shows a fast convergence speed compared with B\&B. 
Moreover, we find that although the RL approach achieves a similar convergence speed, the objective value obtained is lower than our approach. The B\&B approach performs the worst due to inefficiencies in branching.

\begin{figure}[!t]
    \centering
    \includegraphics [scale=0.55,angle=0]{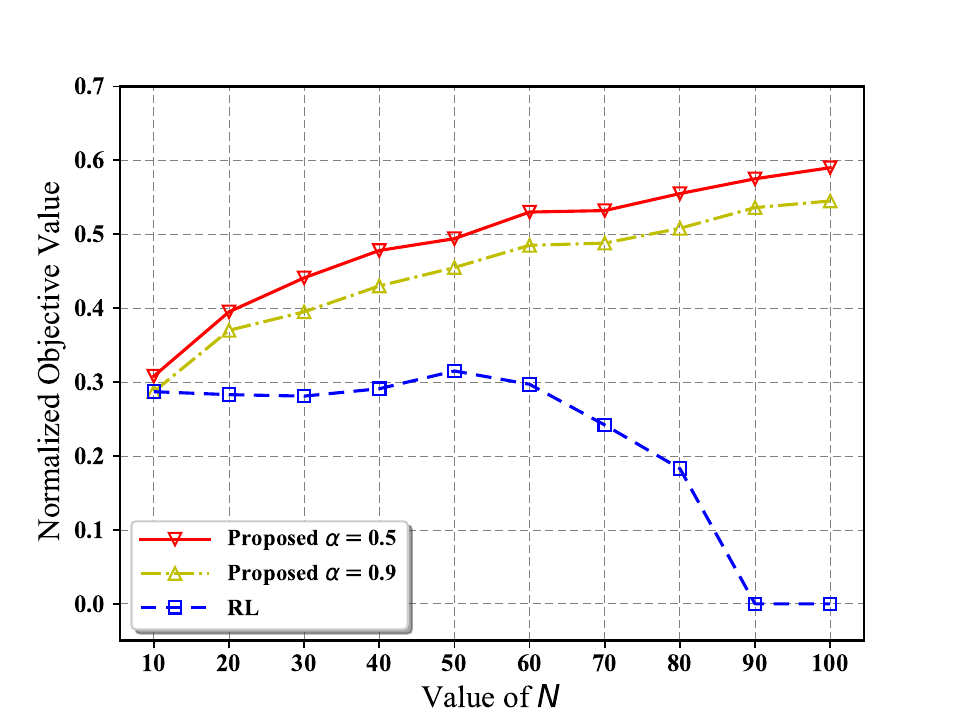}
    \caption{Comparison of normalized objective value for different values of $N$ under \(M = 10\) and \(Q = 0.5\).}
    \label{qos30_p_0.5}
    \includegraphics[scale=0.55]{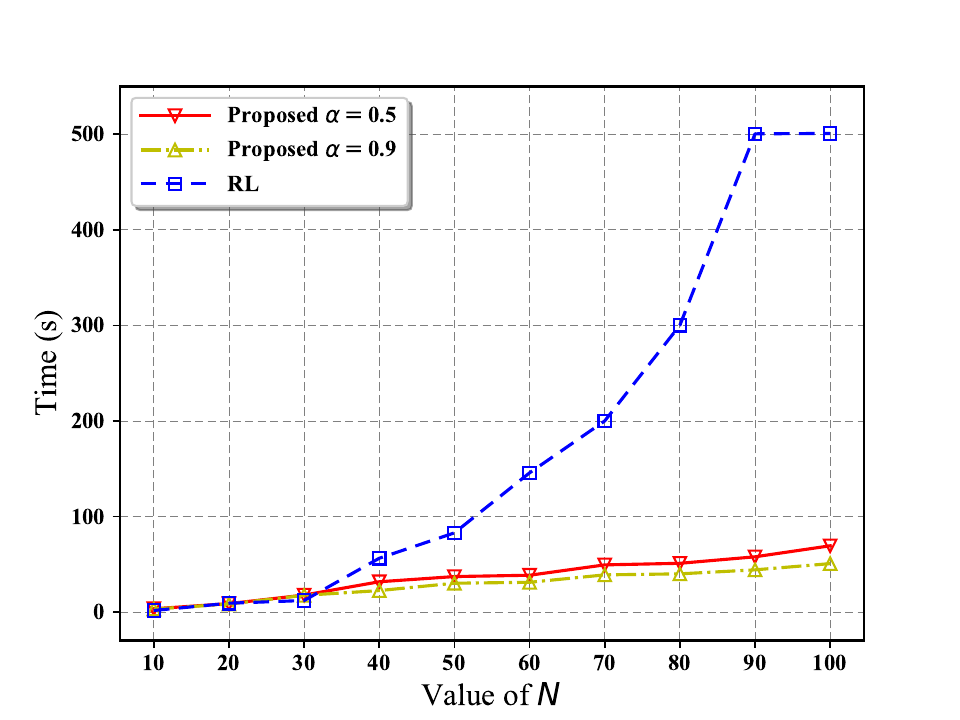}
    \caption{Comparison of running time for different values of $N$ under \(M = 10\) and \(Q = 0.5\).}
    \label{qos30_t_0.5}
\end{figure}


Next, we evaluate the performance for large size problems, where B\&B cannot be used because of complexity. 
Figs. \ref{qos30_p_0.1} and \ref{qos30_t_0.1} show the normalized objective value and running time respectively for different sizes of problem (\ref{S-P2}), with \( Q=0.1 \) in (22a).
In Fig. \ref{qos30_p_0.1}, the proposed approach significantly outperforms the RL approach.
When the problem size increases, the normalized objective value of the RL approach remains constant, while the proposed approach shows a consistent upward trend, particularly when \( \alpha=0.5 \). 
This is because RL tends to get trapped in local optima while exploring feasible solutions, whereas the proposed approach efficiently explores the solution space using relaxed solutions as prior information.
From Fig. 9, it is observed that the running time of the proposed approach is slightly higher than that of the RL approach. 
Although the running time of the proposed approach increases by about 30\% compared to RL when \(N\ge 60\), the computational overhead is acceptable considering the 50\% improvement in the objective value.

Figs. \ref{qos30_p_0.5} and \ref{qos30_t_0.5} show the normalized objective value and running time respectively for different sizes of problem (\ref{S-P2}), with \( Q=0.5 \) in the constraint (22a).
In Fig. \ref{qos30_p_0.5}, the proposed approach maintains a high normalized objective value, whereas the RL approach exhibits a significant performance decline. 
Fig. \ref{qos30_t_0.5} shows that, as the problem size increases, the time needed to execute the RL process increases significantly, far exceeding that of the proposed approach. 
The normalized objective value of the RL approach decreases to zero as \(N\) approaches 90, and the running time reaches 500 seconds, which is set as the unacceptable time threshold. 
This indicates that RL approach suffers from exploration inefficiency when the constraints become tight and feasible solutions become sparse in the possible solution space. 
In contrast, the proposed approach maintains a stable objective value, with running time remains controlled.

\section{Conclusions}
In this paper, we have proposed an optimization approach that integrates convex optimization with RL in the binary decision-making process to address 0-1 mixed problems. 
Specifically, we have solved a relaxed sub-problem before making decisions and use the relaxed solution as prior information to guide RL search policy. 
We have theoretically proven that the neighborhood of the relaxed solution has a higher probability of containing the suboptimal solutions. 
Hence, relaxed solutions can provide prior information to RL for effective guidance of its searching policy. 
Furthermore, we have discussed how to extend our approach to deal with non-convex objective functions and non-convex constraints, thus generalizing the proposed approach to complex real-world networking scenarios. 
Simulation results validated the superiority of our approach under various scenarios, showing a significant improvement in objective value with a minor compromise in computing complexity.  
In general, this paper proposed a unified approach combining RL with optimization theory to efficiently solve a series of 0-1 mixed problems in the field of networking. 
Our approach bridges the existing gap in methodologies for solving complicated networking problems in future emerging applications. 



%

\bibliographystyle{IEEEtran}
\bibliography{bibfile}




\end{document}